\documentclass[a4paper,oneside,reqno,english]{amsart}
\usepackage[latin9]{inputenc}
\usepackage{tikz}
\usepackage{caption}
\usepackage{amsmath}
\usepackage{bbm}
\usepackage{bm}
\usepackage{babel}
\usepackage{verbatim}
\usepackage{mathtools}
\usepackage{amsthm}
\usepackage{amstext}
\usepackage{amssymb}
\usepackage{amsfonts}
\usepackage{comment}
\usepackage{mathabx}
\usepackage{esint}
\usepackage{braket}
\usepackage[unicode=true,
bookmarks=false,
breaklinks=false,pdfborder={0 0 1},backref=false,colorlinks=false]
{hyperref}
\hypersetup{
	colorlinks,linkcolor=blue,anchorcolor=blue,citecolor=blue}
\usepackage{breakurl}
\usepackage{mathrsfs}
\usepackage{cases}
\makeatletter

\numberwithin{equation}{section}
\numberwithin{figure}{section}
\usepackage{enumitem}		

\theoremstyle{plain}
\newtheorem{thm}{Theorem}[section]
\theoremstyle{plain}
\newtheorem{prop}[thm]{Proposition}
\theoremstyle{remark}
\newtheorem{rem}[thm]{Remark}

\theoremstyle{plain}

\theoremstyle{plain}

\theoremstyle{plain}
\newtheorem{lem}[thm]{Lemma}
\theoremstyle{definition}

\theoremstyle{definition}

\theoremstyle{definition}
\newtheorem{con}{Conjecture}

\newcommand{\real}{\mathbb{R}}
\newcommand{\comp}{\mathbb{C}}
\newcommand{\tr}{\textnormal{Tr}}
\newcommand{\iu}{\textbf{i}}
\newcommand{\un}{\bold{1}}

\newcommand{\B}{B}
\newcommand{\C}{C}
\renewcommand{\P}{B'}
\newcommand{\Q}{C'}
\newcommand{\f}{f}
\newcommand{\F}{F}

\newcommand{\com}{\mathbb{C}}
\begin{document}
\title{Some convexity and monotonicity results of trace functionals}
\author{Haonan Zhang}
\address{Institute of Science and Technology Austria (IST Austria),
	Am Campus 1, 3400 Klosterneuburg, Austria}
\address{Current address:
	Department of Mathematics, University of California, Irvine, CA 92617, USA}
\email{haonanzhangmath@gmail.com}
\maketitle

\begin{abstract}
	In this paper, we prove the convexity of trace functionals
	$$(A,B,C)\mapsto \tr |B^{p}AC^{q}|^{s},$$
	for parameters $(p,q,s)$ that are best possible, where $B$ and $C$ are any $n$-by-$n$ positive definite matrices, and $A$ is any $n$-by-$n$ matrix. We also obtain the monotonicity versions of trace functionals of this type. As applications, we extend some results in \cite{HP12quasi,CFL16some} and resolve a conjecture in \cite{RZ14} in the matrix setting. Other conjectures in \cite{RZ14} will also be discussed. We also show that some related trace functionals are not concave in general. Such concavity results were expected to hold in different problems. 
\end{abstract}

\section{Introduction}
The convexity and monotonicity of trace functionals have been widely studied and admit many applications in mathematical physics and quantum information.  For $n\ge 1$, we use $M_n(\comp)$ to denote the set of all $n$-by-$n$ complex matrices,  $M_n^+(\comp)$ the family of all $n$-by-$n$ positive semi-definite matrices, and $M_n^{++}(\comp)$ the collection of all $n$-by-$n$ positive definite matrices.  $M_n^{\times}(\comp)$ will denote the subset of $M_n(\com)$ consisting of all invertible matrices. We use $\un=\un_n$ to denote the identity matrix in $M_n(\com)$. We use the usual trace $\tr$ over matrix algebras and denote by $\langle A,B\rangle :=\tr [AB^\ast]$ the Hilbert--Schmidt inner product. For any linear map $\phi: M_n(\com)\to M_m(\com)$ we denote by $\phi^\dagger$ its adjoint with respect to $\langle\cdot,\cdot \rangle$. For any matrix $A$ we write $|A|:=(A^\ast A)^{1/2}$. For $p>0$, we write $\|A\|_p:=(\tr|A|^p)^{1/p}$. When $p=\infty$, $\|A\|_{\infty}:=\|A\|$ denotes the operator norm of $A$.

\medskip 

Our first main result is the following triple-convexity theorem:

\begin{thm}\label{thm:joint convexity}
	Fix $n\ge 1$ and $K_1,K_2\in M_n(\com)$. For any $0<p,q\le 1/2$ such that $p+q<1$, and $s\ge 1/(1-p-q)$, the functional 
	\begin{equation*}
	\Psi_{p,q,s}(A,B,C):=\tr |B^{-p}K_1AK_2C^{-q}|^s
	\end{equation*}
	is jointly convex in $(A,B,C)\in M_n(\com)\times M^{++}_n(\com)\times M^{++}_n(\com)$.
\end{thm}

Here the parameters $(p,q,s)\in \real^3$ are best possible; see Remark \ref{rem:optimal} below. This extends some results in \cite{HP12quasi,CFL16some}. Our proof relies on a variational method that is used in \cite{Zhang20}. See Section \ref{sect:convexity}.

\medskip 

Our second main result is a monotonicity form of the above theorem for certain parameters. Recall that a linear map $\phi:M_n(\com)\to M_{m}(\com)$ is {\em Schwarz} if 
\begin{equation*}
\phi(XX^\ast)\le \phi(X)\phi(X)^\ast, \qquad X\in M_n(\com).
\end{equation*}

\begin{thm}\label{thm:conj1UCPTP}
	Fix $m,n\ge 1$. Let $0<p,q\le 1/2$ such that $p+q\le1/2$, and $2\le s\le 1/(p+q)$. Suppose that $B,C\in M^{++}_n(\com)$ and $\phi:M_n(\com)\to M_{m}(\com)$ is any unital map such that $\phi^\dagger$ is Schwarz and $\phi(B),\phi(C)\in M_m^{++}(\com)$. Then for any $A\in M_n(\com)$ we have 
	\begin{equation*}
	\tr |\phi(B)^{-p}\phi(A)\phi(C)^{-q}|^s\le \tr |B^{-p}AC^{-q}|^s.
	\end{equation*}
\end{thm}

On one hand, the trace functionals in Theorem \ref{thm:conj1UCPTP} are similar to those in Theorem \ref{thm:joint convexity}, with the difference that the parameters $(p,q,s)$ in Theorem \ref{thm:conj1UCPTP} are more restrictive. On the other hand, the monotonicity version has many advantages and direct physical applications. It is stronger than the corresponding convexity result, as we shall see below. To prove Theorem \ref{thm:conj1UCPTP}, we use Hadamard three-lines theorem and some recent monotonicity results of Carlen and Müller-Hermes \cite{CMH22schwarz} that extended earlier work of Hiai and Petz \cite{HP12quasi}. Our study of monotonicity version of trace functionals in consideration is motivated by some conjectures of Al-Rashed and Zegarli\'nski \cite{RZ14} in investigating the so-called {\em monotone norms}.  This will be briefly recalled later in the Introduction, and we refer to Section \ref{sect:monotonicity} for the proof of Theorem \ref{thm:conj1UCPTP}.

\medskip

Our third main result is the following non-concavity theorem:

\begin{thm}\label{thm:non-concavity}
	For any $p\neq 0$ and $s>2$, there exist $n\ge 2$ and $K_1,K_2\in M_n(\com)$ such that $\Psi_{p,s}(A)=\tr |K_1 A^p K_2|^{s}$ is not concave in $M^{+}_{n}(\com)$.
\end{thm}

Some interesting applications would follow if the above concavity were true, such as some equality conditions of data processing inequalities for $\alpha-z$ R\'enyi relative entropies \cite{Chehade20} and concavity of some interesting scalar functions that is useful in resolving an open problem on {\em $p$-concavity constant} in \cite{Schechtman95remarks}. However, we show that such concavity can not hold in general. The proof uses again the variational method in \cite{Zhang20} and some non-concavity results of triple-variable trace functionals in \cite{CFL16some}. This will be discussed in Section \ref{sect:non-concavity}.

\medskip

In the remaining part of Introduction, let us recall some conjectures of  Al-Rashed and Zegarli\'nski in \cite{RZ14}, which are one of motivations of this paper. The Hilbert spaces in \cite{RZ14} are infinite-dimensional, but the well-definiteness of some functionals was not well addressed. So we consider matrices here for convenience.  

\smallskip

Recall that a linear map $\phi: M_n(\com)\to M_m(\com)$ is {\em $k$-positive} if $\phi\otimes \text{id}_k:M_n(\com)\otimes M_k(\com)\to M_m(\com)\otimes M_k(\com)$ is positive; $\phi$ is {\em completely positive} if it is $k$-positive for all $k\ge 1$. We say that $\phi$ is {\em trace preserving} in case it preserves the usual trace. Stinespring \cite{Stinespring55} gave a very nice characterization of completely positive maps, which will not be recalled here as it will not be used. Unital $2$-positive maps are Schwarz. See \cite{CMH22schwarz} for more discussions.

\smallskip

 For any $n\ge 1$, $\alpha,\beta\in\real,1\le p< \infty$, and any $P\in M^{++}_n(\com)$, we define
\begin{equation*}
\Lambda_{\alpha,\beta,p}(P,X):=\|P^{\alpha/p}XP^{\beta/p}\|_p^p=\tr|P^{\alpha/p}XP^{\beta/p}|^p,\qquad X\in M_n(\com).
\end{equation*}

 The following conjecture corresponds to \cite[Conjectuer I]{RZ14}. 

\begin{con}\label{conjecture1}
	Suppose $\alpha+\beta=-1$ and $\alpha\in [-1,0]$. Let $2\le p<\infty$. Then for any $m,n\ge 1$, for any $P\in M^{++}_n(\com),X\in M_n(\com)$, and for any unital completely positive trace preserving map $\phi: M_n(\com)\to M_m(\com)$, we have 
	\begin{equation*}
	\Lambda_{\alpha,\beta,p}(\phi(P),\phi(X))\le \Lambda_{\alpha,\beta,p}(P,X).
	\end{equation*} 
\end{con}

They also conjectured \cite[Conjecture II]{RZ14} that the above holds for $p\in[1,2)$:

\begin{con}\label{conjecture2}
	Suppose $\alpha+\beta=-1$ and $\alpha\in [-1,0]$. Let $p\in [1,2)$. Then for any $m,n\ge 1$, for any $P\in M^{++}_n(\com),X\in M_n(\com)$, and for any unital completely positive trace preserving map $\phi: M_n(\com)\to M_m(\com)$, we have  
	\begin{equation*}
	\Lambda_{\alpha,\beta,p}(\phi(P),\phi(X))\le \Lambda_{\alpha,\beta,p}(P,X).
	\end{equation*} 
\end{con}

In \cite{RZ14} the maps $\phi$ are completely positive trace preserving, and the unitality condition was accidentally missed. To see that it can not be true for general completely positive trace preserving maps, consider the linear map ({\em partial trace}) $\phi:M_{2n}(\com)\to M_n(\com)$ given by 
\begin{equation*}
\phi 
\begin{pmatrix*}
A &\ast\\
\ast & B
\end{pmatrix*}
:=A+B.
\end{equation*}
Then it is completely positive trace preserving and for
\begin{equation*}
X'=
\begin{pmatrix*}
X&0\\
0&X
\end{pmatrix*}\qquad\text{and}\qquad
P'=
\begin{pmatrix*}
P&0\\
0&P
\end{pmatrix*},
\end{equation*}
we have $\phi(X')=2X$ and $\phi(P')=2P$. Then
\begin{equation*}
\Lambda_{\alpha,\beta,p}(\phi(P'),\phi(X'))=2^{(\alpha+\beta+p)}\Lambda_{\alpha,\beta,p}(P,X).
\end{equation*} 
Since $\alpha+\beta=-1$, the monotonicity 
\begin{equation*}
\Lambda_{\alpha,\beta,p}(\phi(P),\phi(X))\le \Lambda_{\alpha,\beta,p}(P,X),
\end{equation*}
holds only if $2^{\alpha+\beta+p}=2^{p-1}\le 1$. This is impossible either for $p\ge 2$ or $p\in (1,2)$, so both Conjecture I and II in \cite{RZ14} fail (this does not exclude the case $p=1$ in Conjecture II, but we will explain later that even the weaker Conjecture \ref{conjecture2} is false). 

One may also see this by looking at the homogeneity order of $\phi$, which is $\alpha+\beta+p=p-1$, so Conjecture I and II in \cite{RZ14} cannot be true unless $p=1$. Therefore, we consider Conjectures \ref{conjecture1} and \ref{conjecture2} in the sequel.

\medskip

In \cite[Theorem 2.1]{RZ14}, Al-Rashed and Zegarli\'nski proved that Conjecture I in \cite{RZ14} holds for all $p=2^n,n\ge 1$ using a clever induction argument. In their proof the key inequality \cite[(M.1)]{RZ14} 
\begin{equation*}
\phi(P)^{-\alpha}\le \phi(P^\alpha)^{-1},~~\alpha\in (0,1)
\end{equation*}
is not true for general completely positive trace preserving maps (for example the above partial trace $\phi$); but it holds when $\phi$ is furthermore unital. In fact, \cite[(M.1)]{RZ14}  follows from operator monotone decreasing of $x\mapsto x^{-1}$ and 
\begin{equation}\label{ineq:wrong Jensen}
\phi(P^\alpha)\le \phi(P)^{\alpha},~~\alpha\in (0,1).
\end{equation}
that is an operator Jensen type inequality valid for unital positive maps \cite[Theorem 2.1]{Choi74} (see also  \cite[Proposition 2.7.1]{Bhatia09pdm}. So Al-Rashed and Zegarli\'nski  actually confirmed Conjecture \ref{conjecture1} for $p=2^n,n\ge 1$.

As a corollary of Theorem \ref{thm:conj1UCPTP}, we confirm Conjecture \ref{conjecture1} in the affirmative:

\begin{prop}
	Conjecture \ref{conjecture1} holds for more general unital maps whose dual are Schwarz.
\end{prop}

\begin{proof}
	In fact, take $B=C$ in Theorem \ref{thm:conj1UCPTP} and take $(p,q,s)$ in Theorem \ref{thm:conj1UCPTP} to be $(-\alpha/p,-\beta/p,p)$ in Conjecture \ref{conjecture1}.
\end{proof}

Now let us explain why Conjecture \ref{conjecture2} can not hold. This is because the monotonicity implies the joint convexity of $\Lambda_{\alpha,\beta,p}$, which is not true even for scalars. A direct computation shows that the function $\Lambda_{\alpha,\beta,p}(y,x)=x^{p}y^{-1}$ is jointly convex over $(0,\infty)\times (0,\infty)$ only if $p(p-2)\ge 0$. This is impossible for $p\in [1,2)$. So Conjecture \ref{conjecture2} fails.

To see that the monotonicity implies the joint convexity, consider 
\begin{equation}\label{eq:phi}
\phi\begin{pmatrix*}
A & C\\
D& B
\end{pmatrix*}
:=\frac{1}{2}\begin{pmatrix*}
A+B & C+D\\
C+D& A+B
\end{pmatrix*}
\end{equation}
that is unital completely positive trace preserving over $M_{2n}(\com)$. For any $X_j,P_j\in M_n^{++}(\com),j=1,2$, put 
\begin{equation*}
X:=\begin{pmatrix*}
X_1 & 0\\
0& X_2
\end{pmatrix*}\qquad\text{and}\qquad
P:=\begin{pmatrix*}
P_1 & 0\\
0& P_2
\end{pmatrix*}.
\end{equation*}
Suppose that Conjecture \ref{conjecture2} holds. Then we have the mid-point convexity by applying it to the above $\phi$:
\begin{equation*}
\Lambda_{\alpha,\beta,p}\left(\frac{P_1+P_2}{2},\frac{X_1+X_2}{2}\right)\le \frac{1}{2}\left[\Lambda_{\alpha,\beta,p}\left(P_1,X_1\right)+\Lambda_{\alpha,\beta,p}\left(P_2,X_2\right)\right].
\end{equation*}
By continuity, we obtain the usual joint convexity of $\Lambda_{\alpha,\beta,p}$.

\medskip

We will also discuss another conjecture in \cite{RZ14}. To any $P\in M^{++}_n(\com)$ that has unit trace, we associate a (noncommutative Luxemburg) norm $\|\cdot\|_P$. Then \cite[Conjecture IV]{RZ14} concerns the monotonicity of
\begin{equation*}
\langle A,B\rangle_{D,P}:=\left.\dfrac{d}{ds}\dfrac{d}{dt}\right|_{s=t=0}\|D+sA+tB\|_P^2
\end{equation*}
under completely positive trace preserving maps, where the norm is chosen to be $\|X\|_P:=\Lambda_{\alpha,\beta,p}(P,X)^{1/p}$. Let us consider here only a special case $P=\un$: 

\begin{con}\label{conjecture4}
	Fix $m,n\ge 1$. For $1<p<\infty$, the quadratic form defined by 
	\begin{equation}\label{eq:quadradic form-old}
	\langle A,B\rangle_{D}:=\left.\dfrac{d}{ds}\dfrac{d}{dt}\right|_{s=t=0}\|D+sA+tB\|_p^2
	\end{equation}
	is monotone: for all $A,B,D\in M_n(\com)$ and for all completely positive trace preserving maps $\phi: M_n(\com)\to M_m(\com)$,
	\begin{equation*}
	\langle \phi(A),\phi(A)\rangle_{\phi(D)}\le \langle A,A\rangle_{D}.
	\end{equation*}
\end{con}

By the work of Morozova--Chentsov \cite{MC89russian}, Petz \cite{Petz96metric}, and \cite{Kumagai11extented} on the (extended) monotone metrics, Conjecture \ref{conjecture4} can not hold. However, a similar problem has a positive answer: if we replace the power $2$ in \eqref{eq:quadradic form-old} with $p$, then we have the monotonicity under unital maps whose dual are Schwarz for $1<p<2$. Further discussions of monotonicity for more general trace functionals suggest connections to operator log-convexity studied by Ando and Hiai \cite{AH11log}. See Section \ref{sect:otherconjectures} for details.

\subsection*{Acknowledgements} I am grateful to Boguslaw Zegarli\'nski for asking me the questions in \cite{RZ14} and for helpful communication. I also want to thank Paata Ivanisvili for drawing \cite{Schechtman95remarks} to my attention and for useful correspondence.  Many thanks to the anonymous referee for the valuable comments and for pointing out some errors in an earlier version of the paper. This work is partially supported by the European Union's Horizon 2020 research and innovation programme under the Marie Sk\l odowska-Curie grant agreement No. 754411 and the Lise Meitner fellowship, Austrian Science Fund (FWF) M3337.

\section{Convexity results}
\label{sect:convexity}


In this section we prove Theorem \ref{thm:joint convexity}, for which we need the following lemma that was essentially obtained in \cite[Remark 3.8]{Zhang20}. We recall the proof here for reader's convenience.
\begin{lem}\label{lem:variation}
	Suppose that $r_j, 0\le j\le 3$ are positive numbers such that $1/r_0=\sum_{j=1}^{3}1/r_j$. Then for any $n\ge 1$ and any $A,B,C\in M^{\times}_n(\com)$, we have 
	
	\begin{equation}\label{eq:variational-min}
	\tr|BAC|^{r_0}=\min_{X,Y\in M^{\times}_n(\com)} \left\{\frac{r_0}{r_1}\tr|BX|^{r_1}+\frac{r_0}{r_2}\tr|X^{-1}AY^{-1}|^{r_2}+\frac{r_0}{r_3}\tr|YC|^{r_3}\right\},
	\end{equation}
	and 
	\begin{equation}\label{eq:variational-max}
	\tr|B^{-1}AC^{-1}|^{r_2}=\max_{X,Y\in M^{\times}_n(\com)} \left\{\frac{r_2}{r_0}\tr|XAY|^{r_0}-\frac{r_2}{r_1}\tr|XB|^{r_1}-\frac{r_2}{r_3}\tr|CY|^{r_3}\right\}.
	\end{equation}
\end{lem}

\begin{proof}
	We only prove \eqref{eq:variational-min} here as the proof of \eqref{eq:variational-max} is essentially the same.  For any $A,B,C,X,Y\in M^{\times}_n(\com)$ we have 
	\begin{align*}
	\tr|BAC|^{r_0}
	=\|BAC\|_{r_0}^{r_0}
	\le &\|BX\|^{r_0}_{r_1}\|X^{-1}AY^{-1}\|^{r_0}_{r_2}\|YC\|^{r_0}_{r_3}\\
	\le &\frac{r_0}{r_1}\tr|BX|^{r_1}+\frac{r_0}{r_2}\tr|X^{-1}AY^{-1}|^{r_2}+\frac{r_0}{r_3}\tr|YC|^{r_3},
	\end{align*}
	where we have used H\"older's inequality for $\|\cdot\|_p$ (see \cite[Exercise IV.2.7]{Bhatia97matrixanalysis}), and Young's inequality for scalars. 
	Set $W:=BAC$, and let $W=U|W|$ be the polar decomposition. If we choose
	$$X=AC|W|^{(r_0-r_1)/r_1} \qquad \text{and} \qquad Y=|W|^{(r_0-r_3)/r_3}U^\ast BA,$$ 
	then it is easy to check that 
	\begin{equation*}
	BX=BAC|W|^{(r_0-r_1)/r_1}=U|W|^{r_0/r_1},
	\end{equation*}
	\begin{equation*}
	YC=|W|^{(r_0-r_3)/r_3}U^\ast BAC=|W|^{r_0/r_3},
	\end{equation*}
	and 
	\begin{equation*}
	X^{-1}AY^{-1}=|W|^{(r_1-r_0)/r_1}C^{-1}A^{-1}B^{-1}U|W|^{(r_3-r_0)/r_3}
	=|W|^{r_0/r_2}.
	\end{equation*}
	For such $(X,Y)$, the functional in the bracket of right side of \eqref{eq:variational-min} is equal to
	\begin{equation*}
	\frac{r_0}{r_1}\tr|W|^{r_0}+\frac{r_0}{r_2}\tr|W|^{r_0}+\frac{r_0}{r_3}\tr|W|^{r_0}
	=\tr|W|^{r_0}.
	\end{equation*}
	So \eqref{eq:variational-min} is proved.
\end{proof}

\begin{proof}[Proof of Theorem \ref{thm:joint convexity}]
	By approximation, we may assume that $K_1,K_2$ are invertible. The proof is based on the following variational formula:
	\begin{equation}
	\begin{aligned}\label{eq:variation}
	&\tr |B^{-p}K_1AK_2C^{-q}|^s\\
	=&\max_{X,Y\in M^{\times}_n(\com)}\left\{\frac{s}{r}\tr|XK_1AK_2Y|^r-sp\tr|XB^{p}|^{1/p}-sq\tr|C^{q}Y|^{1/q}\right\},
	\end{aligned}
	\end{equation}
	where $r^{-1}=p+q+s^{-1}$. Let us first finish the proof given \eqref{eq:variation}. It is easy to see that the joint convexity of $\Psi_{p,q,s}$ will follow from the convexity of 
	$$\Psi_1(A):=\tr|XK_1AK_2Y|^r,\qquad A\in M^{\times}_n(\com),$$
	for any $K_1,K_2,X,Y\in M^{\times}_n(\com)$, and the concavity of 
	$$\Psi_2(B):=\tr|XB^{p}|^{1/p},\qquad B\in M^{+}_{n}(\com),$$
	and 
	$$\Psi_3(C):=\tr|C^{q}Y|^{1/q},\qquad C\in M^{+}_{n}(\com),$$
	for any $X,Y\in M^{\times}_n(\com)$. Indeed, this is a consequence of \eqref{eq:variation} and the fact that convexity is stable under taking maximum (see \cite[Lemma 3.2]{Zhang20} for a proof). 
	
	By assumption, $r=(p+q+s^{-1})^{-1}\ge 1$. So $\|XK_1\cdot K_2Y\|_r$ is convex by Minkowski inequality. Thus $\Psi_1$, as the composition of scalar convex function $x^r,r\ge 1$ and $\|XK_1\cdot K_2Y\|_r$, is also convex. 
	Since $0<p,q\le 1/2$, the functionals
	\begin{equation*}
	\Psi_2(B)=\tr(XB^{2p}X^\ast)^{1/(2p)}\qquad \text{and} \qquad \Psi_3(C):=\tr(Y^*C^{2q}Y)^{1/(2q)}
	\end{equation*}
	are both concave for any $X,Y\in M^{\times}_n(\com)$. This is due to a result of Epstein \cite{Epstein73}. See \cite{CFL18survey,Zhang20} for further discussions. 
	
	Now it remains to prove the variational formula \eqref{eq:variation}. This is a direct consequence of Lemma \ref{lem:variation}. In fact, taking in \eqref{eq:variational-max}
	$$(r_0,r_1,r_2,r_3)=(r,1/p,s,1/q)$$
	and
	$$(A,B,C)=(K_1\tilde{A}K_2,\tilde{B}^{p},\tilde{C}^{q}),$$
	we get
	\begin{align*}
	&\tr|\tilde{B}^{-p}K_1\tilde{A}K_2 \tilde{C}^{-q}|^{s}\\
	=&\max_{X,Y\in M^{\times}_n(\com)} \left\{\frac{s}{r}\tr|XK_1\tilde{A}K_2Y|^{r}-sp\tr|X\tilde{B}^{p}|^{1/p}-sq\tr|\tilde{C}^{q}Y|^{1/q}\right\},
	\end{align*}
	which is exactly \eqref{eq:variation}.
\end{proof}

\begin{rem}
	The joint convexity/concavity of two-variable trace functionals 
	$$M^{++}_n(\com)\times M^{++}_n(\com)\ni (B,C)\mapsto\Psi_{p,q,s}(\un,B,C)$$ 
	have been understood very well \cite{Zhang20}. For the triple-variable case, the joint convexity of 
	$$M_n(\com)\times M^{++}_n(\com)\times M^{++}_n(\com)\ni(A,B,C)\mapsto \tr[A^\ast B^pA C^r]$$
	was known when $p,r<0$ and $-1\le p+r<0$. See \cite[Example 6]{HP12quasi} or \cite[Corollary 3.3]{CFL16some}. This corresponds to Theorem \ref{thm:joint convexity} with $s=2$ and $K_1=K_2=\un$. 
\end{rem}
\begin{rem}\label{rem:optimal}
	The range for $(p,q,s)$ in Theorem \ref{thm:joint convexity} is optimal in the following sense: if for non-zero (otherwise it is already known \cite{Zhang20}) $\alpha,\beta$ and $\gamma$, the functional
	\begin{equation*}
	f(A,B,C)=\tr|B^{\alpha}AC^{\beta}|^\gamma,\qquad A,B,C \in M_n^{++}(\com),
	\end{equation*}
	is jointly convex for any $n\ge 1$, then necessarily
	$$-\frac{1}{2}\le \alpha,\beta< 0, \alpha+\beta>-1 \text{  and  } \gamma\ge\frac{1}{1+\alpha+\beta}.$$
	In fact, it is known \cite[Theorem 1.1, Proposition 2.3]{Zhang20} that for $p,q\in \real\setminus\{0\}, s>0$, the functional 
	$$(A,B)\mapsto \tr(B^{q/2}A^p B^{q/2})^s,\qquad A,B \in M_n^{++}(\com),$$
	is jointly convex for any $n\ge 1$ if and only if either 
	$$-1\le p,q<0, \qquad s>0$$
	or 
	$$-1\le \min\{p,q\}<0, \qquad 1\le \max\{p,q\}\le 2, \qquad p+q>0, \qquad s\ge 1/(p+q).$$
	Using this, we deduce that
	\begin{enumerate}
		\item $\gamma>0$ and $-1/2\le \alpha <0$ by joint convexity of $ f(A,B,\un)=\tr(B^{\alpha}A^2B^{\alpha})^{\gamma/2}$;
		\item $-1/2\le \beta<0$ by joint convexity of $(A,C)\mapsto f(A,\un,C)=\tr(C^{\beta}A^2C^{\beta})^{\gamma/2}$;
		\item $\alpha+\beta>-1$ and $\gamma\ge 1/(1+\alpha+\beta)$ by joint convexity of $f(A,A,C)=\tr(C^{\beta}A^{2+2\alpha}C^{\beta})^{\gamma/2}$.
	\end{enumerate}
\end{rem}

\section{Monotonicity under unital maps whose dual are Schwarz}
\label{sect:monotonicity}
In this section we prove the monotonicity of a family of triple-variable trace functionals under unital maps whose dual are Schwarz. For this we need a few lemmas. Let $S$ be the open strip $\{z\in\comp:0< \Re z< 1\}$ and $\bar{S}$ be its closure. 

\begin{lem}\label{lem:three line}
	Fix $n\ge 1$ and $B,C\in M^{+}_n(\com)$. Let $F:\bar{S}\to M_n(\com)$ be a bounded function that is analytic on $S$ and continuous on $\bar{S}$. Suppose that $1\le p_0<p_1\le \infty$ and
	\begin{equation*}
	M_k:=\sup_{t\in\real}\|\B^{1/p_k}F(k+\iu t)\C^{1/p_k}\|_{p_k}<\infty, \qquad k=0,1.
	\end{equation*}
	Then for any $\theta\in (0,1)$ and $1/p_{\theta}:=(1-\theta)/p_0+\theta/p_1$, we have
	\begin{equation*}
	\|\B^{1/p_{\theta}}F(\theta)\C^{1/p_{\theta}}\|_{p_{\theta}}\le M_0^{1-\theta}M_1^{\theta}.
	\end{equation*}
	
\end{lem}

\begin{proof}
	For any $1\le p\le \infty$, we denote by $p'$ the conjugate index of $p$, i.e. $1/p+1/p'=1$. Then we have by duality
	\begin{equation}\label{eq:duality}
	\|\B^{1/p_{\theta}}F(\theta)\C^{1/p_{\theta}}\|_{p_{\theta}}
	=\sup\left\{|\langle \B^{1/p_{\theta}}F(\theta)\C^{1/p_{\theta}}, X^\ast\rangle|:\|X\|_{p'_{\theta}}\le 1 \right\}.
	\end{equation}
	Now take any $X$ such that $\|X\|_{p'_{\theta}}\le 1$. Let $X=U|X|$ be the polar decomposition with $U$ being unitary. Put $X(z):= U|X|^{p'_{\theta}/p'_z}$, where
	\begin{equation}\label{eq:p_z}
	\frac{1}{p_z}:=\frac{1-z}{p_0}+\frac{z}{p_1}, \qquad z\in\comp,
	\end{equation}
	or equivalently
	\begin{equation*}
	\frac{1}{p'_z}=\frac{1-z}{p'_0}+\frac{z}{p'_1}, \qquad z\in\comp.
	\end{equation*}
	Then by definition, we have $X(\theta)=U|X|=X$. For any $x\in [0,1]$ and $t\in\real$, we may write $X(x+\iu t)=U_t|X|^{p'_\theta/p'_x}$ for some unitary $U_t$, and thus
	\begin{equation}\label{ineq:p-norm bound of X}
	\|X(x+\iu t)\|_{p'_x}=\|X\|_{p'_{\theta}}^{p'_{\theta}/p'_x}\le 1.
	\end{equation}
	Consider the function 
	\begin{equation*}
	G(z):=\langle \B^{1/p_z}F(z)\C^{1/p_z},X(z)^\ast\rangle,
	\end{equation*}
	which is analytic on $S$ and continuous on $\bar S$. Clearly, $G(\theta)=\langle \B^{1/p_{\theta}}F(\theta)\C^{1/p_{\theta}}, X^\ast\rangle$. By H\"older's inequality and \eqref{ineq:p-norm bound of X}, we have for any $x\in [0,1]$ and $t\in\real$ that
	\begin{align*}
	|G(x+\iu t)|
	\le &\| V_t\B^{1/p_x}F(x+\iu t)\C^{1/p_{x}}W_t\|_{p_x} \|X(x+\iu t)\|_{p'_x}\\
	\le &\| \B^{1/p_x}\|_{2p_x}\|F(x+\iu t)\|_\infty\|\C^{1/p_{x}}\|_{2p_x}\\
	=&\| \B\|_{2}^{1/p_x}\|F(x+\iu t)\|_\infty\|\C\|_{2}^{1/p_x},
	\end{align*}
	where $V_t$ and $W_t$ are two unitaries. By assumption, this is bounded from above uniformly in $z=x+\iu t\in\bar{S}$. So $G$ is bounded in $\bar{S}$. A similar argument gives
	\begin{align*}
	|G(k+\iu t)|
	\le \|\B^{1/p_k}F(k+\iu t)\C^{1/p_{k}}\|_{p_k} \|X(k+\iu t)\|_{p'_k}
	\le  M_k
	\end{align*}
	for $k=0,1$ and $t\in\real$.	Applying Hadamard three-lines theorem to $G$, we get 
	\begin{equation*}
	|G(\theta)|=|\langle \B^{1/p_{\theta}}F(\theta)\C^{1/p_{\theta}},X^\ast\rangle|\le M_0^{1-\theta}M_1^{\theta}.
	\end{equation*}
	Since this holds for any $X$ such that $\|X\|_{p'_\theta}\le 1$, we conclude the proof by \eqref{eq:duality}.
\end{proof}

\begin{lem}\label{lem:Riesz-Thorin}
	Fix $m,n\ge 1$. Let $\B,\C\in M^{++}_n(\com)$ and $\P,\Q\in M^{++}_m(\com)$. Suppose that $1\le p_0\le p_1\le \infty$ and $1\le q_0\le q_1\le \infty$. For $\theta\in [0,1]$, put
	\begin{equation*}
	\frac{1}{p_{\theta}}:=\frac{1-\theta}{p_0}+\frac{\theta}{p_1}\qquad \text{and}\qquad \frac{1}{q_{\theta}}:=\frac{1-\theta}{q_0}+\frac{\theta}{q_1}.
	\end{equation*} 
	Suppose that $T:M_n(\com)\to M_m(\com)$ is a bounded linear operator satisfying
	\begin{equation*}
	\|\P^{1/q_k}T(X)\Q^{1/q_k}\|_{q_k}
	\le M_k \|\B^{1/p_k}X\C^{1/p_k}\|_{p_k}, \qquad k=0,1,
	\end{equation*}
	for all $X\in M_n(\com)$, where $M_k<\infty,k=0,1$. Then we have
	\begin{equation*}
	\|\P^{1/q_\theta}T(X)\Q^{1/q_\theta}\|_{q_\theta}
	\le M_0^{1-\theta}M_1^{\theta} \|\B^{1/p_\theta}X\C^{1/p_\theta}\|_{p_\theta},\qquad X\in M_n(\com).
	\end{equation*}
\end{lem}

\begin{proof}
	Let $\B^{1/p_\theta}X\C^{1/p_\theta}=U|\B^{1/p_\theta}X\C^{1/p_\theta}|$ be the polar decomposition. Consider 
	\begin{equation*}
	X(z):=\B^{-1/p_z}U|\B^{1/p_\theta}X\C^{1/p_\theta}|^{p_\theta/p_z}\C^{-1/p_z},
	\end{equation*}
	where $p_z$ is defined as in \eqref{eq:p_z}. Since $T$ is bounded and linear, the function $F(z):=T(X(z))$ is analytic on $S$ and continuous on $\bar S$. By definition, for $x\in[0,1]$ and $t\in\real$,
	\begin{equation}\label{eq:unitaries in X}
	X(x+\iu t)=\B^{-1/p_x}U_t|\B^{1/p_\theta}X\C^{1/p_\theta}|^{p_\theta/p_x}V_t\C^{-1/p_x},
	\end{equation}
	where $U_t$ and $V_t$ are two unitaries. So we have
	\begin{equation*}
	\|X(x+\iu t)\|\le \left(\|\B^{-1}\|\|\B\| \|X\|^{p_\theta}\|\C^{-1}\|\|\C\|\right)^{1/p_x},
	\end{equation*}
	which is bounded from above uniformly in $z=x+\iu t\in\bar{S}$. Since $T$ is bounded, $F(z)=T(X(z))$ is bounded on $\bar{S}$. By assumption, for $k=0,1$ and $t\in\real$,
	\begin{align*}
	\|\P^{1/q_k}F(k+\iu t)\Q^{1/q_k}\|_{q_k}
	\le& M_k \|\B^{1/p_k}X(k+\iu t)\C^{1/p_k}\|_{p_k}\\
	= &M_k \|U_t|\B^{1/p_\theta}X\C^{1/p_\theta}|^{p_\theta/p_k}V_t\|_{p_k}\\
	=&M_k \|\B^{1/p_\theta}X\C^{1/p_\theta}\|_{p_\theta}^{p_\theta/p_k},
	\end{align*}
	where $U_t$ and $V_t$ are  two unitaries in \eqref{eq:unitaries in X} for $x=k$. Now apply Lemma \ref{lem:three line} to $F$, and we finish the proof.
\end{proof}

\begin{lem}\label{lem:casep=2}
	Fix $m,n\ge 1$. Let $0< \alpha,\beta< 1 $ and $\alpha+\beta\le 1$. Suppose that $A,B\in M^{++}_n(\com)$ and $\phi:M_n(\com)\to M_m(\com)$ is a unital map such that $\phi^\dagger$ is Schwarz and $\phi(A),\phi(B)\in M^{++}_m(\com)$. Then for any $X\in M_n(\com)$, we have 
	\begin{equation}\label{ineq:casep=2}
	\tr \left[\phi(X)^\ast\phi(A)^{-\alpha}\phi(X)\phi(B)^{-\beta} \right]
	\le \tr\left[X^\ast A^{-\alpha} X B^{-\beta}\right],
	\end{equation}
	for all $X\in M_n(\com)$.

\end{lem}

\begin{proof}
	Note first that it suffices to prove \eqref{ineq:casep=2} for $0< \alpha,\beta< 1 $ and $\alpha+\beta=1$. In fact, fix $0<a,b\le 1$. Since $\phi$ is unital, an operator Jensen type inequality for unital positive maps \cite[Theorem 2.1]{Choi74} (see also \cite[Proposition 2.7.1]{Bhatia09pdm}) gives $\phi(B^{a})\le \phi(B)^{a}$. Note that $x\mapsto x^{-b}$ is operator monotone decreasing, we get $\phi(B)^{-ab}\le \phi(B^{a})^{-b}$. Choose $(a,b)=(\beta/(1-\alpha),1-\alpha)$ and we deduce $\phi(B)^{-\beta}\le \phi(B^{\beta/(1-\alpha)})^{-(1-\alpha)}$. So
	\begin{equation*}
	\tr \left[\phi(X)^\ast\phi(A)^{-\alpha}\phi(X)\phi(B)^{-\beta} \right]
	\le \tr\left[\phi(X)^\ast\phi(A)^{-\alpha}\phi(X)\phi(B^{\beta/(1-\alpha)})^{-(1-\alpha)} \right],
	\end{equation*}
	which will be bounded from above by the right hand side of \eqref{ineq:casep=2} if we could prove \eqref{ineq:casep=2} for $0< \alpha,\beta< 1 $ such that $\alpha+\beta=1$.
	
	The $\alpha+\beta=1$ case follows from \cite[Theorem 11]{CMH22schwarz}. 
\end{proof}

Now we are ready to prove Theorem \ref{thm:conj1UCPTP}.

\begin{proof}[Proof of Theorem \ref{thm:conj1UCPTP}]
	Applying Lemma \ref{lem:Riesz-Thorin} to $$(\B,\C,\P,\Q)=(B^{-ps},C^{-qs},\phi(B)^{-ps},\phi(C)^{-qs}),$$
	we know that 
	the desired result will follow from 
	\begin{equation}\label{ineq:r}
	\|\phi(B)^{-ps/r}\phi(A)\phi(C)^{-qs/r}\|_{r}\le \|B^{-ps/r}AC^{-qs/r}\|_{r}
	\end{equation} 
	for $r=2$ and $r=\infty$, since $2\le s<\infty$. When $r=2$, it becomes:
	\begin{equation*}
	\tr \left[\phi(A)^\ast\phi(B)^{-\alpha}\phi(A)\phi(C)^{-\beta} \right]
	\le \tr\left[A^\ast B^{-\alpha} A C^{-\beta}\right],
	\end{equation*}
	with $\alpha=ps$ and $\beta=qs$. By assumption, $0<\alpha,\beta<1$ and $\alpha+\beta\le 1$. So \eqref{ineq:r} holds for $r=2$ by Lemma \ref{lem:casep=2}.
	The inequality \eqref{ineq:r} for $r=\infty$ is a consequence of Russo--Dye Theorem (see \cite[Theorem 2.3.7]{Bhatia09pdm} and \cite{RD66duke}), since $\phi$ is unital, linear and positive. So the proof is finished. 
\end{proof}

\begin{rem}
	The argument here is similar to Beigi's proof of data processing inequality for certain sandwiched R\'enyi relative entropies \cite{beigi2013sandwiched}. The endpoint case $r=2$ here is non-trivial. 
\end{rem}

\section{Non-concavity and non-convexity results}
\label{sect:non-concavity}
There are several reasons of expecting the concavity of 
\begin{equation}\label{eq:defn of Psi}
\Psi(A):=\tr |K_1 A^p K_2|^{1/p}, \qquad A\in M^{+}_n(\com)
\end{equation}
for $0<p<1/2$ and any $K_1,K_2\in M^{+}_n(\com)$. In \cite{Chehade20} Chehade proposed a strategy to obtain equality conditions of data processing inequality for $\alpha-z$ R\'enyi relative entropies, where the concavity of \eqref{eq:defn of Psi} is the key. For more about the $\alpha-z$ R\'enyi relative entropies and their data processing inequalities, see \cite{AD15alpha-z,Zhang20} and references therein. 

 Another application of this concavity result (if it were true) is the concavity of 
\begin{equation}\label{concavity-unconditional}
\real_+^n\ni (u_1,\dots,u_n)\mapsto |\sum_{j=1}^{n} u_j^{p}r_j|^{1/p}
\end{equation}
for any $\{r_j\}_{j=1}^{n}\subset\real$ and any $n\ge 1$. To see this, just take $K_1=\sum_{j=1}^{n}\ket{1}\bra{j}$, $A=\sum_{j=1}^{n}u_j\ket{j}\bra{j}$ and $K_2=\sum_{j=1}^{n}r_j\ket{j}\bra{1}$. The concavity of \eqref{concavity-unconditional} will imply the concavity of 
\begin{equation*}
\real_+^n\ni (u_1,\dots,u_n)\mapsto \int|\sum_{j=1}^{n} u_j^{p}r_j|^{1/p}.
\end{equation*}
Here $r_j$'s are Rademacher functions. So one may prove that the $1/p$-concavity constant of the Rademacher sequence in $L_{1/p}$ is $1$, which was left open in \cite{Schechtman95remarks}. 

One evidence of supporting the concavity of \eqref{eq:defn of Psi} is the case when $K_1=K_2^\ast$ (or $K_1=\un$). In this case we know well the concavity of $\Psi$. See \cite{Zhang20} or the following Remark \ref{remark:K1 K2} for example. Unfortunately and surprisingly, for general $K_1$ and $K_2$, $\Psi$ is not concave anymore. 

More generally we have the non-concavity Theorem \ref{thm:non-concavity} as stated in the Introduction. 

\begin{proof}[Proof of Theorem \ref{thm:non-concavity}]
	We are not going to find any particular examples to disprove the concavity. The proof is based on some concavity and non-concavity results that are combined via the variational Lemma \ref{lem:variation}.
	
	Recall that $p\ne 0$ and $s>2$. Set $(r_0,r_2):=(2,s)$. Since $s>2$, we can choose $2<r_1,r_3<\infty$ such that $1/r_0=\sum_{j=1}^{3}1/r_j$. Here the exact values of $r_1$ and $r_3$ are not important. Now we apply the variational formula \eqref{eq:variational-min} to obtain
	\begin{equation}\label{eq:varional-contradiction}
	\tr| B^{q} A^{p} C^{r}|^{2}=\min_{X,Y\in M^{\times}_n(\com)} \left\{\frac{r_0}{r_1}\tr|B^{q} X|^{r_1}+\frac{r_0}{r_2}\tr|X^{-1}A^{p} Y^{-1}|^{r_2}+\frac{r_0}{r_3}\tr|YC^{r}|^{r_3}\right\},
	\end{equation} 
	where $(q,r)=(1/r_1,1/r_3)$.
	
	Since $0<q=1/r_1,r=1/r_3< 1/2$, the functionals 
	\begin{equation*}
	B\mapsto \tr|B^{q} X|^{r_1}=\tr(X^\ast B^{2q} X)^{1/(2q)}
	\end{equation*}
	and 
	\begin{equation*}
	C\mapsto \tr|YC^r|^{r_3}=\tr(Y C^{2r} Y^\ast)^{1/(2r)}
	\end{equation*}
	are concave on $M^{+}_n(\com)$ for any $X,Y\in M_n(\com)$ and any $n\ge 1$. See again \cite{Epstein73,CFL18survey,Zhang20} for proofs and related results.
	
	If the assertion of this theorem is false, i.e. $\Psi_{p,s}$ is concave for any $K_1,K_2\in M_n(\com)$ and any $n\ge 2$, then the functional
	\begin{equation*}
	A\mapsto \tr|X^{-1}A^{p} Y^{-1}|^{r_2}=\tr|X^{-1}A^{p} Y^{-1}|^{s}
	\end{equation*}
	is concave for any $X,Y\in M^{\times}_n(\com)$. All combined, we deduce the joint concavity of 
	\begin{equation*}
	\Phi(A,B,C):=\tr| B^{q} A^{p} C^{r}|^{2}=\tr [A^{p} B^{2q}A^{p} C^{2r}],
	\end{equation*} 
	for such (non-zero) $p,q,r$. Here we used \eqref{eq:varional-contradiction} and the fact that concavity is stable under taking minimum (see \cite[Lemma 3.2]{Zhang20} for a proof). However, by a result of Carlen, Frank and Lieb \cite[Corollary 3.3]{CFL16some}, $\Phi$ is never concave in $(A,B,C)$ for non-zero $p,q,r$. This leads to a contradiction and completes the proof. 
\end{proof}

\begin{rem}\label{remark:K1 K2}
	When $K_2=K_1^\ast$ and $s>0$, we know that $\Psi_{p,s}(A)=\tr(K_1 A^p K_1^\ast)^s$ is concave if and only if $0<p\le 1$ and $s\le 1/p$ \cite{Hiai13,CFL18survey,Zhang20}. So the most interesting part of this theorem is the case $0<p<1/2$ and $s>2$.
\end{rem}

\begin{rem}
	Here we argued by contradiction. In \cite{vershynina2022convexity} another proof is given using concrete counterexamples. 
\end{rem}

With the non-convexity results of triple-variable functionals in \cite{CFL16some}, we can derive the following non-convexity result of one-variable functionals in a similar way. This time we only state the most interesting part. 

\begin{thm}\label{thm:non-convexity}
	For any $1/2\le p<1$ and $1/p\le s<2$, there exist $n\ge 2$ and $K_1,K_2\in M_n(\com)$ such that $\Psi_{p,s}(A)=\tr |K_1 A^p K_2|^{s}$ is not convex.
\end{thm}

\begin{proof}
	Since $1/s\in (1/2,p]\subset (1/2,1]$, we can find $r_1,r_3>0$ such that $1/r_1,1/r_3\in (0,1/2]$ such that $1/s=\sum_{j=1}^{3}1/r_j$ with $r_2=2$. Let $q=1/r_1$ and $r=1/r_3$. By \eqref{eq:variational-max}, we have
	\begin{equation}
	\begin{aligned}\label{eq:variation-nonconvex}
	\tr |B^{-q}A^p C^{-r}|^{r_2}
	=\max_{X,Y\in M^{\times}_n(\com)}\left\{\frac{r_2}{s}\tr|XA^pY|^{s}-\frac{r_2}{r_1}\tr|XB^{q}|^{r_1}-\frac{r_2}{r_3}\tr|C^{r}Y|^{r_3}\right\}.
	\end{aligned}
	\end{equation}
	Since $0<q,r\le 1/2$, the functionals (\cite{Epstein73,CFL18survey,Zhang20})
	\begin{equation*}
	B\mapsto \tr|XB^{q}|^{r_1}=\tr(X B^{2q} X^\ast)^{1/(2q)}
	\end{equation*}
	and 
	\begin{equation*}
	C\mapsto \tr|C^rY|^{r_3}=\tr(Y^\ast C^{2r} Y)^{1/(2r)}
	\end{equation*}
	are concave in $M^{+}_n(\com)$ for any $X,Y\in M_n(\com)$ and any $n\ge 1$. If the assertion is false, i.e. $\Psi_{p,s}(A)=\tr |K_1 A^p K_2|^{s}$ is convex for all $K_1,K_2\in M_n(\com)$ and all $n\ge 2$, then the functional 
	\begin{equation*}
	A\mapsto \tr|XA^pY|^{s}
	\end{equation*}
	is convex for any $X,Y\in M^{\times}_n(\com)$. All combined, we deduce the joint convexity of 
	\begin{equation*}
	\Phi(A,B,C):=\tr| B^{-q} A^{p} C^{-r}|^{2}=\tr [A^{p} B^{-2q}A^{p} C^{-2r}].
	\end{equation*}
	However, by a result of Carlen, Frank and Lieb \cite[Corollary 3.3]{CFL16some}, $\Phi$ is jointly convex only if $p=1$, $q,r>0$ and $q+r\le 1$. This leads to a contradiction and completes the proof. 
\end{proof}

\section{On Conjecture \ref{conjecture4}}
\label{sect:otherconjectures}
We first explain why Conjecture \ref{conjecture4} can not hold in general. This is based on work of Morozova--Chentsov \cite{MC89russian}, Petz \cite{Petz96metric}, and \cite{Kumagai11extented} on the (extended) monotone metrics. Recall that \cite[Definition 2.2]{Kumagai11extented} 
$$K:\cup_{n=1}^{\infty}(M_n(\com)\times M_n(\com)\times M_n^{++}(\com))\to \com,\quad (A,B,D)\mapsto K_D(A,B)$$
is an {\em extended monotone metric} if the following conditions are satisfied:
\begin{enumerate}
	\item $(A,B)\mapsto K_D(A,B)$ is sesquilinear for every $D\in M_n^{++}(\com)$.
	\item $K_D(A,A)\ge 0$ and the equality holds if and only if $A=0$.
	\item $D\mapsto K_D(A,A)$ is continuous on $M_n^{++}(\com)$ for every $A\in M_n(\com)$.
	\item $K_{\phi(D)}(\phi(A),\phi(A))\le K_D(A,A)$ for every completely positive trace preserving map $\phi: M_n(\com)\to M_m(\com)$, $D\in M_n^{++}(\com)$, $A\in M_n(\com)$ and $m,n\ge 1.$ 
\end{enumerate}

According to \cite[Theorem 4.1]{RZ14} and its proof, 
\begin{align*}
\langle A,B\rangle_D=&(1-\frac{p}{2})\tr [|G|^{p-2}(G^\ast A+A^\ast G)]\cdot \tr [|G|^{p-2}(G^\ast B+B^\ast G)]\\
&+\tr [|G|^{p-2}(A^\ast B+B^\ast A)]+\frac{p-2}{2}\int_{0}^{1}ds\int_{0}^{\infty}dt \\
&\cdot \tr\left[\frac{|G|^{s(p-2)}}{t+|G|^2}(G^\ast B+B^\ast G)\frac{|G|^{(1-s)(p-2)}}{t+|G|^2}(G^\ast A+A^\ast G)\right]
\end{align*}
and 
\begin{align*}
\langle A,A\rangle_D
\ge \tr [|G|^{p-2}(A^\ast A+A^\ast A)]
\ge 0,
\end{align*}
where $G=D/\|D\|_p$.
From this, we see that the conditions (1-3) are satisfied. 

\begin{thm}
	Conjecture \ref{conjecture4} fails. 
\end{thm}

\begin{proof}
Suppose that Conjecture \ref{conjecture4} is true. Then the above discussion shows that $(A,B,D)\mapsto \langle A,B\rangle_D$ is an extended monotone metric.

Kumagai \cite[Theorem 3.2]{Kumagai11extented} characterized extended monotone metric, and we recall here only the classical version \cite[Lemma 3.1]{Kumagai11extented} due to Campbell \cite{campbell1986extended}: there exists a uniquely a pair of a continuous function $b:\real_+\to \real$ and a continuous positive function $c:\real_+\to \real_+$ such that 
\begin{equation*}
\langle A,B\rangle_D=b(\tr D)(\tr A)^\ast (\tr B)+(\tr D)c(\tr D)\tr[D^{-1}A^\ast B],
\end{equation*}
where $A,B\in M_n(\com)$ and $D\in M_n^{++}(\com)$ are mutually commutative.  So when $A,B,D\in M_n^{++}(\com)$ are all diagonal:
\begin{align*}
&2(2-p)\|D\|_p^{2-2p}\tr[D^{p-1}A] \tr[D^{p-1}B]+2(p-1)\|D\|_p^{2-p}\tr [D^{p-2}AB]\\
=&b(\tr D)(\tr A)^\ast (\tr B)+(\tr D)c(\tr D)\tr[D^{-1}A^\ast B].
\end{align*}
Choosing $A=B=D\in M_n^{++}(\com)$, we get
\begin{equation*}
2\|D\|_p^2=(\tr D)^2[ b(\tr D)+c(\tr D)].
\end{equation*}
When $n\ge 2$, we may find $D_1,D_2\in M_n^{++}(\com)$ such that $\tr D_1= \tr D_2=x$ while $2\|D_1\|_p^2=y\neq z=2\|D_2\|_p^2$. Then we deduce
\begin{equation*}
y=x^2[b(x)+c(x)]=z,
\end{equation*}
which leads to a contradiction. This finishes the proof. 
\end{proof}

However, a problem similar to Conjecture \ref{conjecture4} has a positive answer: the monotonicity under unital maps whose dual are Schwarz holds if we replace the square with $p$-th power for $1<p<2$. This was essentially obtained in \cite{Li20sobolev} and let us discuss it further. For any function $\F:(0,\infty)\times (0,\infty)\to (0,\infty)$, and for any $A,B\in M^{++}_n(\com)$ with spectral decompositions $A=\sum_{j}\lambda_j E^A_j$ and $B=\sum_{k}\mu_k E^B_k$, put
\begin{equation}
Q_{\F}^{A,B}(X):=\sum_{j,k}\F(\lambda_j,\mu_k) E^A_j X E^B_k
=F(L_A,R_B)X,
\end{equation}
where $L_A$ and $R_B$ are the left and right multiplication operators, respectively. 
For any $\f\in C^2(0,\infty)$, consider the trace functional 
\begin{equation*}
\Xi_{\f}(D,A):=\left.\dfrac{d^2}{ds^2}\right|_{s=0}\tr f(D+sA)=\left.\dfrac{d}{ds}\right|_{s=0}\tr[f'(D+sA)A].
\end{equation*}
A direct computation shows
\begin{align*}
\Xi_{\f}(D,A)
=&\left\langle A,Q^{D,D}_{\f'^{[1]}}(A)\right\rangle,
\end{align*}
where $g^{[1]}$ denotes the difference quotient of $g$:
\begin{equation*}
g^{[1]}(s,t):=\begin{cases*}
\frac{g(s)-g(t)}{s-t}& $s\neq t$\\
g'(s)&$s=t$
\end{cases*}.
\end{equation*}

\begin{lem}\label{lem:direct computation}
	Let $F:(0,\infty)\times (0,\infty)\to (0,\infty)$.
	\begin{enumerate}
		\item 	For $D_{ij}\in M_{n}^{++}(\com), i,j=1,2$, $A_k\in M_{n}(\com), k=1,2$, and
		\begin{equation}\label{eq:defn of D1 D2 A}
		D_1=\begin{pmatrix*}
		D_{11} & 0\\
		0& D_{12}
		\end{pmatrix*},
		\quad 
		D_2=\begin{pmatrix*}
		D_{21} & 0\\
		0& D_{22}
		\end{pmatrix*},
		\quad
		A=\begin{pmatrix*}
		A_{1} & 0\\
		0& A_{2}
		\end{pmatrix*},
		\end{equation}
		we have 
		\begin{equation}\label{eq:sum}
		\left\langle Q^{D_1,D_2}_F(A),A\right\rangle = \left\langle Q^{D_{11},D_{21}}_F(A_1),A_1\right\rangle + \left\langle Q^{D_{12},D_{22}}_F(A_2),A_2\right\rangle.
		\end{equation}
		\item For $D_1,D_2\in M_{n}^{++}(\com), X\in M_n(\com)$, and 
		\begin{equation}\label{eq:defn of AD}
		D=\begin{pmatrix*}
		D_{1} & 0\\
		0& D_{2}
		\end{pmatrix*},
		\quad 
		A=\begin{pmatrix*}
		0 & X\\
		0& 0
		\end{pmatrix*},
		\end{equation}
		we have 
		\begin{equation}\label{eq:symmetry}
		\left\langle Q^{D,D}_F(A),A\right\rangle = \left\langle Q^{D_{1},D_{2}}_F(X),X\right\rangle.
		\end{equation}
	\end{enumerate}

\end{lem}

\begin{proof}
	Both (1) and (2) follow from the following identity: 
	\begin{equation}\label{eq:general}
	Q_F^{D_1,D_2}
	\begin{pmatrix*}
	A_{11}&A_{12}\\
	A_{21}&A_{22}
	\end{pmatrix*}
	=\begin{pmatrix*}
	Q^{D_{11},D_{21}}_F(A_{11})&Q^{D_{11},D_{22}}_F(A_{12})\\
	Q^{D_{12},D_{21}}_F(A_{21})&Q^{D_{12},D_{22}}_F(A_{22})
	\end{pmatrix*},
	\end{equation}
	where $D_{ij},A_{ij}\in M_n^{++}(\com), i,j=1,2$ and 
		\begin{equation*}
	D_1=\begin{pmatrix*}
	D_{11} & 0\\
	0& D_{12}
	\end{pmatrix*},
	\quad 
	D_2=\begin{pmatrix*}
	D_{21} & 0\\
	0& D_{22}
	\end{pmatrix*}.
	\end{equation*}
	The proof of \eqref{eq:general} is a direct computation. Let $D_{lm}=\sum_{j} \lambda^{(lm)}_{j}E^{lm}_j$ be the spectral decomposition. Then by definition, we have
	\begin{align*}
	&Q_F^{D_1,D_2}
	\begin{pmatrix*}
	A_{11}&A_{12}\\
	A_{21}&A_{22}
	\end{pmatrix*}\\
	=&\sum_{j,k}F(\lambda^{(11)}_{j},\lambda^{(21)}_{k})
	\begin{pmatrix*}
	E^{11}_j &0\\
	0&0
	\end{pmatrix*}
	\begin{pmatrix*}
	A_{11}&A_{12}\\
	A_{21}&A_{22}
	\end{pmatrix*}
	\begin{pmatrix*}
	E^{21}_k &0\\
	0&0
	\end{pmatrix*}	\\
	&+\sum_{j,k}F(\lambda^{(11)}_{j},\lambda^{(22)}_{k})
	\begin{pmatrix*}
	E^{11}_j &0\\
	0&0
	\end{pmatrix*}
	\begin{pmatrix*}
	A_{11}&A_{12}\\
	A_{21}&A_{22}
	\end{pmatrix*}
	\begin{pmatrix*}
	0 &0\\
	0&E^{22}_k
	\end{pmatrix*}	\\
	&+\sum_{j,k}F(\lambda^{(12)}_{j},\lambda^{(21)}_{k})
	\begin{pmatrix*}
	0 &0\\
	0&E^{12}_j
	\end{pmatrix*}
	\begin{pmatrix*}
	A_{11}&A_{12}\\
	A_{21}&A_{22}
	\end{pmatrix*}
	\begin{pmatrix*}
	E^{21}_k &0\\
	0&0
	\end{pmatrix*}	\\
	&+\sum_{j,k}F(\lambda^{(12)}_{j},\lambda^{(22)}_{k})
	\begin{pmatrix*}
	0 &0\\
	0&E^{12}_j
	\end{pmatrix*}
	\begin{pmatrix*}
	A_{11}&A_{12}\\
	A_{21}&A_{22}
	\end{pmatrix*}
	\begin{pmatrix*}
	0 &0\\
	0&E^{22}_k
	\end{pmatrix*}	\\
	=&\sum_{j,k}F(\lambda^{(11)}_{j},\lambda^{(21)}_{k})
	\begin{pmatrix*}
	E^{11}_j A_{11}E^{21}_k &0\\
	0&0
	\end{pmatrix*}
	+\sum_{j,k}F(\lambda^{(11)}_{j},\lambda^{(22)}_{k})
	\begin{pmatrix*}
	0 &E^{11}_j A_{12}E^{22}_k\\
	0&0
	\end{pmatrix*}\\
	&+\sum_{j,k}F(\lambda^{(12)}_{j},\lambda^{(21)}_{k})
	\begin{pmatrix*}
	0 &0\\
	E^{12}_j A_{21}E^{21}_k&0
	\end{pmatrix*}
	+\sum_{j,k}F(\lambda^{(12)}_{j},\lambda^{(22)}_{k})
	\begin{pmatrix*}
	0 &0\\
	0&E^{12}_j A_{22}E^{22}_k
	\end{pmatrix*}\\
	=&\begin{pmatrix*}
	Q^{D_{11},D_{21}}_F(A_{11})&Q^{D_{11},D_{22}}_F(A_{12})\\
	Q^{D_{12},D_{21}}_F(A_{21})&Q^{D_{12},D_{22}}_F(A_{22})
	\end{pmatrix*}.
	\end{align*}
\end{proof}

\begin{thm}\label{thm:double operator}
	Let $f\in C^2(0,\infty)$ be such that $f''$ is continuous positive function on $(0,\infty)$.
	Consider the following statements:
	\begin{enumerate}
		\item For all unital maps $\phi: M_n(\com)\to M_m(\com)$ whose dual $\phi^\dagger$ are Schwarz, all $(D_1,D_2,A)\in M_n^{++}\times M_n^{++}(\com)\times M_n(\com)$ and all $m,n\ge 1$, we have
		\begin{equation*}
			\left\langle \phi(A),Q^{\phi(D_1),\phi(D_2)}_{\f'^{[1]}}(\phi(A))\right\rangle\le \left\langle A,Q^{D_1,D_2}_{\f'^{[1]}}(A)\right\rangle.
		\end{equation*}
		\item The function $(A,D_1,D_2)\mapsto \left\langle A,Q^{D_1,D_2}_{\f'^{[1]}}(A)\right\rangle$ is jointly convex in $M_n(\com)\times M^{++}_n(\com)\times M^{++}_n(\com)$.
		\item The function $(A,D)\mapsto \left\langle A,Q^{D,D}_{\f'^{[1]}}(A)\right\rangle$ is jointly convex in $M_n(\com)\times M^{++}_n(\com)$.
		\item $f'$ is operator concave, or equivalently \cite[Theorem 3.3]{BS1955monotone}, for any $x_0\in (0,\infty)$, the function 
		\begin{equation*}
		x\mapsto \frac{f'(x)-f'(x_0)}{x-x_0}
		\end{equation*}
		is operator monotone decreasing.
		\item $f'':(0,\infty)\to (0,\infty)$ is operator monotone decreasing, or equivalently (proof of main theorem in \cite{hansen2006trace} or \cite[Theorem 3.1]{AH11log}), it admits the integral representation
		\begin{equation}\label{eq:integral rep}
		f''(x)=a+\int_{[0,\infty)}\frac{\lambda+1}{\lambda+x}d\mu(\lambda),
		\end{equation}
		where $a\ge 0$ and $\mu$ is some finite Borel measure on $[0,\infty)$.
	\end{enumerate}
	 Then (5) $\implies $ (1) $\implies$ (2) $\Longleftrightarrow$ (3) $\implies$ (4).
\end{thm} 

\begin{rem}
	For $f(x)=x^p,1<p<2$, we have 
	\begin{equation*}
	f''(x)=p(p-1)x^{p-2}=\frac{p(p-1)\sin (p-1)\pi}{\pi}\int_{0}^{\infty}\frac{\lambda^{p-2}}{\lambda+x}d\lambda.
	\end{equation*}
	So \eqref{eq:integral rep} holds with $a=0$ and 
	$$d\mu(\lambda)=\frac{p(p-1)\sin (p-1)\pi}{\pi}\cdot\frac{\lambda^{p-2}}{\lambda+1}d\lambda.$$
\end{rem}

\begin{proof}[Proof of Theorem \ref{thm:double operator}]
	To show  (5) $\implies $ (1), note that by the integral representation \eqref{eq:integral rep}: for $x,y>0$:
\begin{align*}
f'(x)-f'(y)
&=a(x-y)+\int_{y}^{x}\int_{[0,\infty)}\frac{\lambda+1}{\lambda+t}d\mu(\lambda)dt\\
&=a(x-y)+\int_{[0,\infty)}(\lambda+1)(\log(\lambda+x)-\log(\lambda+y))d\mu(\lambda).
\end{align*}
This, together with the integral representation
\begin{equation*}
\frac{\log x-\log y}{x-y}=\int_{0}^{\infty}\frac{1}{(x+s)(y+s)}ds,
\end{equation*}
yields
\begin{equation*}
\frac{f'(x)-f'(y)}{x-y}=a+\int_{0}^{\infty}\int_{[0,\infty)}\frac{\lambda+1}{(\lambda+x+s)(\lambda+y+s)}d\mu(\lambda)ds.
\end{equation*}
Therefore,
\begin{align*}
&\left\langle Q^{D_1,D_2}_{\f'^{[1]}}(A),A\right\rangle=a\tr[AA^\ast]\\
 &\qquad \qquad \qquad +\int_{[0,\infty)}(\lambda+1)\int_{0}^{\infty}\tr\left[\frac{1}{\lambda+s+D_1}A\frac{1}{\lambda+s+D_2}A^\ast \right]dsd\mu(\lambda).
\end{align*}
and thus
\begin{align*}
&\left\langle Q^{\phi(D_1),\phi(D_2)}_{\f'^{[1]}}(\phi(A)),\phi(A)\right\rangle=a\tr[\phi(A)\phi(A)^\ast]\\
&\qquad \qquad  +\int_{[0,\infty)}(\lambda+1)\int_{0}^{\infty}\tr\left[\frac{1}{\lambda+s+\phi(D_1)}\phi(A)\frac{1}{\lambda+s+\phi(D_2)}\phi(A)^\ast \right]dsd\mu(\lambda).
\end{align*}
So to prove (5) $\implies $ (1) it remains to show 
\begin{equation*}
\tr[\phi(A)\phi(A)^\ast]\le \tr[AA^\ast],
\end{equation*}
and for all $\lambda >0$:
\begin{align*}
&\int_{0}^{\infty}\tr\left[\frac{1}{\lambda+s+\phi(D_1)}\phi(A)\frac{1}{\lambda+s+\phi(D_2)}\phi(A)^\ast \right]ds\\
&\qquad \qquad \qquad \qquad \le \int_{0}^{\infty}\tr\left[\frac{1}{\lambda+s+D_1}A\frac{1}{\lambda+s+D_2}A^\ast \right]ds.
\end{align*}
The former is equivalent to 
\begin{equation*}
\tr[\phi^\dagger(B)\phi^\dagger(B)^\ast]\le \tr[BB^\ast],
\end{equation*}
in view of $\phi^\dagger\phi\le \text{id}_n \Leftrightarrow \phi\phi^\dagger\le \text{id}_m$. Since $\phi$ is unital and $\phi^\dagger$ is Schwarz:
\begin{equation*}
\tr[\phi^\dagger(B)\phi^\dagger(B)^\ast]\le \tr[\phi^\dagger(BB^\ast)]=\tr[BB^\ast].
\end{equation*}
The latter follows from a monotonicity version of a theorem of Lieb \cite[Theorem 11]{CMH22schwarz} that extends the results in \cite{HP12quasi}. In fact, since $\phi$ is unital,
\begin{align*}
&\int_{0}^{\infty}\tr\left[\frac{1}{\lambda+s+\phi(D_1)}\phi(A)\frac{1}{\lambda+s+\phi(D_2)}\phi(A)^\ast \right]ds\\
&\qquad \qquad \qquad \qquad =\int_{0}^{\infty}\tr\left[\frac{1}{s+\phi(\lambda+D_1)}\phi(A)\frac{1}{s+\phi(\lambda+D_2)}\phi(A)^\ast \right]ds.
\end{align*}
Then applying \cite[Theorem 11]{CMH22schwarz} to $f(x)=\int_{0}^{1}x^sds=\frac{x-1}{\log x}$ (see also \cite[Example 5]{HP12quasi}), we have for Schwarz map $\phi^\dagger$ that
\begin{align*}
&\int_{0}^{\infty}\tr\left[\frac{1}{s+\phi(\lambda+D_1)}\phi(A)\frac{1}{s+\phi(\lambda+D_2)}\phi(A)^\ast \right]ds\\
&\qquad \qquad \qquad \qquad  \le \int_{0}^{\infty}\tr\left[\frac{1}{\lambda+s+D_1}A\frac{1}{\lambda+s+D_2}A^\ast \right]ds.
\end{align*}
This finishes the proof of  (5) $\implies $ (1).
 
 Now we show $(1) \implies (2)$. Denote $g(D_1,D_2,A):=\left\langle A,Q^{D_1,D_2}_{\f'^{[1]}}(A)\right\rangle$. Consider the map $\phi$ given in \eqref{eq:phi}. Then it is unital and its dual is Schwarz. Now for any 
 $$(D_{11}, D_{12},A_1), (D_{21},D_{22},A_2)\in M^{++}_n(\com)\times M^{++}_n(\com)\times M_n(\com),$$
 we form the operators $(D_1,D_2,A)\in M^{++}_{2n}(\com)\times M^{++}_{2n}(\com)\times M_{2n}(\com)$ as in \eqref{eq:defn of D1 D2 A}. Then by \eqref{eq:sum} we have
 $$g(\phi(D_1),\phi(D_2),\phi(A))=2g\left(\frac{D_{11}+D_{12}}{2},\frac{D_{21}+D_{22}}{2},\frac{A_1+A_2}{2}\right)$$
 and 
 $$g(D_1,D_2,A)=g(D_{11},D_{21},A_1)+g(D_{12},D_{22},A_2).$$
So (1) implies the mid-point joint convexity, and thus the usual joint convexity by continuity, of $g$. This proves (2).
 
The implication $(2)\implies (3)$ is trivial. The converse follows from the following identity (see \eqref{eq:symmetry})
\begin{equation*}
\left\langle A,Q^{D_1,D_2}_{\f'^{[1]}}(A)\right\rangle=\left\langle X,Q^{D,D}_{\f'^{[1]}}(X)\right\rangle,
\end{equation*}
with 
\begin{equation*}
D=\begin{pmatrix*}
D_1& 0\\
0&D_2
\end{pmatrix*}
\qquad \text{and} \qquad
X=\begin{pmatrix*}
0& A\\
0&0
\end{pmatrix*}.
\end{equation*}

It remains to show $(3)\implies (4)$. For this take $D_2=x_0 \un$. Then $Q^{D_1,D_2}_{f'^{[1]}}=h(L_{D_1})$ and thus 
$$\left\langle A,Q^{D_1,D_2}_{f'^{[1]}}(A)\right\rangle=\left\langle A,h(L_{D_1})A\right\rangle,$$
where $h(x)=(f'(x)-f'(x_0))/(x-x_0)$ and $L_{D_1}(X):=D_1X$. By a standard argument \cite{HP12quasi}, $h$ is operator monotone decreasing. In fact, for any $\xi\in \com^n$ take $A=A_\xi\in M_n(\com)$ as the matrix with $\xi$ being the first column and zeros elsewhere. Then we have
$$\left\langle A_\xi,Q^{D_1,D_2}_{f'^{[1]}}(A_\xi)\right\rangle=\langle A_\xi,h(L_{D_1})A_\xi\rangle=\left\langle \xi, h(D_1)\xi\right\rangle, $$
where the last $\langle\cdot,\cdot \rangle$ is the scalar product on $\com^n$. It is known that \cite[Theorem 3.1]{AH11log} the joint convexity of $\langle \xi, h(D_1)\xi\rangle$ in $(D_1,\xi)\in M^{++}_n(\com)\times \com^n$ implies that $h$ is operator monotone decreasing. See \cite[Theorem 1.45]{Carlen22review} for a short proof. 
\end{proof}

\end{document}